\documentclass[11pt]{amsart}

\usepackage{amsmath,amsfonts,latexsym,amssymb,amscd, graphicx}
\usepackage{latexsym,enumerate,epsfig,listings}
\usepackage{amsthm,amsopn,amstext}
\usepackage[latin1]{inputenc}
\usepackage{verbatim}
\usepackage{graphicx}
\usepackage{hyperref}
\usepackage[dvipsnames, usenames]{color}
\usepackage{anysize}
\usepackage{float}
\usepackage{color}

\newcommand{\yy}{{\mathbf y}}

\newcommand{\xx}{{\mathbf x}}

\newcommand{\esp}{{ \mathbb E}}

\newcommand{\argmin}{\mathop{\rm arg\,min}}

\newtheorem{thm}{Theorem}

\newtheorem{lem}{Lemma}
\newtheorem{prop}{Proposition}

\numberwithin{equation}{section}

\begin{document}

\title{Bus operators in competition: a directed location approach}

\author[Fernanda Herrera]{Fernanda Herrera $\dag$ }
\author[Sergio I. L\'opez]{Sergio I. L\'opez*}

\address[ $\dag$ ]{School of Global Policy and Strategy\\ University of California San Diego\\ CA 92093}
\address[*]{Departamento  de Matem\'aticas\\ Facultad de Ciencias, UNAM\\
	C.P. 04510, Ciudad de M\'exico, M\'exico}

\email[ $\dag$ ]{fherrera@ucsd.edu}
\email[*]{silo@ciencias.unam.mx}

\thanks{Fernanda Herrera gratefully acknowledges financial support from the University of California Institute for Mexico and the United States (UC MEXUS). Sergio I. L\'opez was partially funded by Conacyt-SNI 215989 grant.}
\date{\today}

\subjclass{C710, C720, C730, R410}

\begin{abstract}
We present a directed variant of Salop's (1979) model to analyze bus transport dynamics. Players are operators competing in both cooperative and non-cooperative games. Utility, like in most bus concession schemes in emerging countries, is proportional to the total fare collection. Competition for picking up passengers leads to well documented and dangerous driving practices that cause road accidents, traffic congestion and pollution. We obtain theoretical results that support the existence and implementation of such practices, and give a qualitative description of how they come to occur. In addition, our results allow to compare the base transport system with a more cooperative one.
\end{abstract}

\maketitle

\section{Introduction} 
In this work, we model the competition of bus operators for passengers in a public transport concession scheme. The models -which are directed variants of the Salop model \cite{Salop}, in turn a circuit adaptation of the classic Hotelling model \cite{Hot}- are a characterization of Mexico City's transport system. According to a 2017 survey, 74.1\% of the trips made in Mexico City by public transport are carried out on buses with concession contracts \cite{EOD2017}.

Much like in other Latin American cities, the contracts that lay the responsibilities, penalties and service areas, are rarely enforced by the corresponding authorities, and \textit{in these instances, the main driver determining the planning and operations tend to be the operator's profit margins} \cite{Comparative} p. 9. Leaving the task to companies or even drivers themselves, has lead to what \cite{Foster} refer to as \textit{curious old practices}: driving habits adopted by bus operators, whose salary is proportional to the fare collection, to maximize the number of users boarding the unit. While these practices were observed and recorded in the United Kingdom in the 1920s, they are very much present today, particularly in cities with emerging economies and sub optimal concession plans. The practices enlisted in \cite{Foster} pertaining to driving are:

\begin{enumerate}
	\item \textit{Hanging back or Crowling}. Operators drive slowly to pick up as many people as possible. The idea is that long waiting times increase the number of passengers waiting at stops. A variant is to stop altogether until the bus is fully loaded, or the next bus catches up.
	\item \textit{Racing}. When an operator deems that the number of passengers waiting at a stop is not worth making the stop. In this case, she continues driving in the hopes of collecting more users ahead.

	\item \textit{Overtaking, Tailing or Chasing}. 
	Attempting to pass the bus ahead, to cut in and pick up the passengers frontwards.
	
	\item \textit{Turning}. When an empty or nearly empty bus turns around before the end of the route, and drives back in the opposite direction. 
\end{enumerate}

Many of these practices have negative consequences on the service provided to users, and as a byproduct, on the perception of public transport. In the 2019 survey on victimization in public transport \cite{EV2019}, carried out in Mexico City and its metropolitan area, 50\% of the interviewees deemed the quality of concession transport to be bad, and 15\% very bad. Moreover, 27\% considered that traveling in concession transport was somewhat dangerous, and 60\% very dangerous. In both of these dimensions, concession public transport did worse than any other form of transport, including public and private types. The matter is pressing enough that the current administration of Mexico City stressed in its Strategic Mobility Plan of 2019 \cite{PEM2019} p. 9: \textit{The business model that governs this (transport) sector, (...) produces competition in the streets for users, which results in the pick up and drop off of passengers in unauthorized places, increased congestion and a large number of traffic incidents each year.}

A solution to these problems may involve the deregulation of public transport to increase competition between providers, and to create incentives for providing a differentiated product, namely better service in the form of shorter waiting times, and safer driving practices. As an example, Margaret Thatcher introduced the Transport Act 1985 \cite{Acta}, which lead to the privatization of bus services, higher competition between companies, and a set of norms to abide by, like keeping vehicles in good condition, avoiding dangerous driving, and establishing routes and publishing timetables. However successful, this type of measure seems unlikely for Mexico and other developing Latin American countries, both for legislative reasons and corruption in the implementation. So, with this work we aim to shed light on the implications of a transport system where operators compete for passengers without regulation.

To be specific, we model the situation where bus operators compete to maximize their utility, which is proportional to the number of passengers boarding the units. As a proxy for the number of passengers collected, we use the road ahead up to the next bus. The strategies available to drivers are the driving speeds. Time, like strategies themselves, is continuous. For simplicity, we do not allow drivers to change speed any time they want, instead we assume that they maintain a chosen speed for a given time, and let them change in the next. While practical, the assumption also reflects the empirical observation that bus drivers make strategic stops along the road, where they obtain information on the game. More precisely, they pay agents that collect the arrival times of previous buses to that particular stop, and even the identity of the drivers themselves. This way, the operators realize whether they are competing against known drivers, and more importantly, whether they changed their speed. With this information, they make their decision for the next part of the route. We obtain a simple interpretation of the results that is consistent with the driving practices mentioned above.

To the best of our knowledge, our approach is novel, and it allows us to model a variety of scenarios and obtain explicit descriptions of equilibria. Furthermore, we are able to explore the time evolution of the adopted strategies. All the results are expressed in terms of the behavior of the operators. Given the tractability of our models, some natural theoretical questions emerge.

Relevant literature on transport problems includes \cite{Newell} modeling of the optimal headway bus service from the point of view of a central dispatcher. In the historical context of Transport Act 1985 \cite{Acta}, several scientific articles analyzed the effect of the privatization. Under the assumption of the existence of an economic equilibrium in the competition system, \cite{Foster} classify the driving practices into two categories: those consistent with the equilibrium, and those who are not. They analyzed the expected timetables in the deregulated scenario. In \cite{Evans} a comparative analysis of fare and timetable allocation in competition, monopoly and net benefit maximization (both restricted and unrestricted to a zero profit) models is presented. Building on from this, \cite{Ireland} introduces the consumer's perspective and obtains the equilibria prices and number of services offered by transport companies. The possibility of predatory behavior between two enterprises competing through fares and service level, is analyzed by \cite{Dodgson}, using the data from the city of Inverness. In \cite{Ellis} the authors study the optimal policies of competing enterprises in terms of fares, and the bus service headway, in a unique bus stop and destination scenario. They also introduce the concept of \textit{demand coordination} which can be implemented through timetables. Assuming a spatial directed model with a single enterprise, \cite{dePalma} finds the timetable that minimizes the costs associated to service delays. The work of \cite{Borenstein} analyzes flight time data and finds empirical evidence to support Hotelling models. From a non-economic perspective, \cite{Baik} models competing buses in a circuit behaving like random particles with repulsion between them (meaning they could not pass each other). A contemporary review on transport market models using game theory is given by \cite{Adler}, and a general review of control problems which arise in buses transport systems is presented in \cite{Ibarra}.\\

This paper is organized as follows. In Section \ref{Sec:Model}, we present the general model, and the single and two player games. Relevant definitions, notation and interpretations are introduced. In Section \ref{Sec:Results} we present the solutions to the games and include in Subsection \ref{subsect:time} the evolution of the strategies adopted by the operators. That is, we look at the long-run equilibria of the games. We also introduce a natural extension of the two player games and present the results in Subsection \ref{sub:ext}. Concluding remarks are in \ref{Sec:Conclu}, and proofs are in Appendix \ref{app}. 

\section{The model}\label{Sec:Model}

The assumptions of the game are the following. There are $n\leq2$ buses, each is driven by one of $n$ operators along a route. There is only one type of bus and one type of driver, meaning that the buses have identical features, and that the drivers are homogeneous in terms of skill and other relevant characteristics. 

The speed of a bus, denoted by $v$, is bounded throughout every time and place of the road by:
\begin{equation}\label{speedbound}
	0 < v_{min}  \leq v \leq v_{max} , 
\end{equation}
where the constants $v_{min}$ and $v_{max}$ are fixed, and determined by exogenous factors like the condition of the bus, Federal and State laws and regulations, the infrastructure of the road, etc.

Drivers can pick up passengers along any point on the route at any given time. In other words, there are no designated bus stations, nor interval-based time schedules in place. This scenario is an approximation to a route with a large number of homogeneously distributed bus stops.

We allow for infinite bus capacity, so drivers can pick up any number of passengers they come across. Alternatively, one can assume that passengers alight from the bus almost right after boarding it, so the bus is virtually empty and ready to pick up users at any given time. The important point to note is that passengers that have boarded a bus will not hop on the next, either because they never descended it in the first place, or because they already reached their final destination if they did.

Bus users reach their pick up point at random times, so demand for transport is proportional to the time elapsed between bus arrivals. Let $\lambda >0$ denote the mean number of passengers boarding a bus per unit of time, and let $p \geq 0$ denote the fixed fare paid by each user. We assume that there is a fixed driving cost $c \geq 0$ per unit of time. This cost summarizes fuel consumption, maintenance, protection insurance for the bus and passengers, etc. 

The operators get a share of the total revenue, and consequently seek to maximize it. Since they cannot control the number of passengers on the route, the fare, or the driving costs, the only resource available to them is to set the driving speed, which we assume remains constant throughout the time interval $[0,T]$, with $T>0$. The strategy space of a bus driver is then
\begin{equation}\label{space}
	\Gamma=  \{  v \geq 0  :  v_{min} \leq  v \leq v_{max}  \} ,
\end{equation}
where $v_{min}$ and $v_{max}$ are given in \eqref{speedbound}. We define a mixed strategy, $X$ or $Y$, to be a random variable taking values in the space $\Gamma$. 

In what follows we define the expected utility of drivers given a set of assumptions on the number of players and their starting positions, the fixed variables of the models, and route characteristics. Relevant notation and concepts are introduced when deemed necessary.

\subsection{Single player games} We first consider a game with only one driver picking up passengers along the road. Importantly, the fact that only one bus is covering the route implies that commuters have no option but to wait for its arrival, the player is aware of this.

\begin{itemize}
	\item \textbf{Fixed-distance game} 
	
	A single bus departs the origin of a route of length $D$. We adopt the convention that the initial time is whenever the bus departs the origin. We define the expected utility of driving at a given speed $v$ to be
	\begin{equation}\label{1-dist-uti}
		u(v):= (p \lambda) T  - c T,
	\end{equation}
	where $T= \frac{D}{v} $ is the time needed to travel the distance $D$ at speed $v$. \\
	
	Note that since there is no other bus picking up passengers, the expected number of people waiting for the bus in a fixed interval of the road increases proportionally with time. From this, one infers that the expected total number of passengers taking the bus is proportional to the time it takes the bus to reach its final destination.\footnote{This justifies the first summand in \eqref{1-dist-uti}.} The conclusion and its implication can be expressed rigorously using a space-time Poisson process, see for example \cite{Gelf} pp. 283-296.\\
	
	\item \textbf{Fixed-time game}
	
	Suppose now that the driver chooses a constant speed $v$ satisfying \eqref{speedbound} in order to drive for $T$ units of time. The bus then travels the distance $D=Tv$, which clearly depends on $v$. We define the expected utility of driving at a given speed $v$ to be
	\begin{equation}\label{1-time-uti}
		u(v):= (p \lambda) D - c T.
	\end{equation}
	The underlying assumption is that for sufficiently small $T$, there are virtually no new arrivals of commuters to the route, so effectively, the number of people queuing for the bus remains the same as that of the previous instant. The requirement is that $T$ is small compared to the expected interarrival times of commuters.
	
	It follows that the total amount of money collected by the driver is proportional to the total distance traveled by the bus.
	
\end{itemize}

\subsection{Two player games} There are two buses picking up passengers along a route, which we assume is a one-way traffic circuit. An advantageous feature of circuits is that buses that return from any point on the route to the initial stop may remain in service; this is generally not the case in other types of routes. In particular, we assume that the circuit is a one-dimensional torus of length $D$. For illustration purposes and without loss of generality, from now on we require the direction of traffic to be clockwise.

We define the $D$-module of any real number $r$ as
$$(r)_{mod \, D} := \frac{r}{D} - \Big \lfloor \frac{r}{D} \Big \rfloor , $$
where $\lfloor z \rfloor$ is the greatest integer less than or equal to $z$.

The interpretation of $(r)_{mod \, D}$ is the following: if starting from the origin, a bus travels the total distance $r$, then $(r)_{mod \, D}$ denotes its relative position on the torus. Indeed, $r$ may be such that the bus loops around the circuit many times, nonetheless $(r)_{mod \, D}$ is in $[0,D)$ for all $r$. We refer to $r$ as the absolute position of the bus, and to $(r)_{mod \, D}$ as the relative position (with respect to the torus). Note that the origin and the end of the route share the same relative position, since $(0)_{mod \, D}=0=(D)_{mod \, D}$.

Let $\xx$ and $\yy$ denote the two players of the game, and let $x, y$ be their respective relative positions. The directed distance function $d_{\xx}$ is given by
\begin{eqnarray}\label{def:distance}
	d_{\xx}(x,y) := \left\{ \begin{array}{ll}
		y - x &  \textrm{ if } x \leq y, \\
		D+ y - x & \textrm{ if } x > y .
	\end{array} \right. 
\end{eqnarray}

Equation \eqref{def:distance} has a key geometrical interpretation: it gives the distance from $x$ to $y$ considering that traffic is one-way. The interest of this is that the potential amount of commuters $\xx$ picks up is proportional to the distance between $x$ and $y$, namely $d_{\xx}(x,y)$. See Figure \ref{Fig:directed}.

A straightforward observation is that for any real number $r$, we have 
\begin{equation}\label{loopinv}
	d_{\xx}( (x+r)_{mod \, D}, (y+r)_{mod \, D} ) =d_{\xx}(x,y) . 
\end{equation}
This asserts that if we shift the relative position of the two players by $r$ units (either clockwise or counterclockwise, depending on the sign of $r$), then the directed distance $d_{\xx}$ is unchanged.

One can define the directed distance $d_{\yy}$ analogously,
\begin{displaymath}
	d_{\yy}(x,y) := \left\{ \begin{array}{ll}
		x - y &  \textrm{ if } y \leq x, \\
		D+ x - y & \textrm{ if } y > x .
	\end{array} \right. 
\end{displaymath}

By definition, there is an intrinsic symmetry between $d_{\xx}$ and $d_{\yy}$: we have $d_{\xx}(x,y) = d_{\yy}(y,x)$ and $d_{\yy}(x,y) = d_{\xx}(y,x)$. Roughly speaking, this means that if we were to swap all the labels, namely $\xx$ to $\yy$, $x$ to $y$,\footnote{Importantly, this switches the relative positions of the players.} and vice versa, then it suffices to plug the new labels into the previous definitions to obtain the directed distances.

Another immediate observation is that for any pair of different positions $(x,y)$, the sum of the two directed distances gives the total length of the circuit,
\begin{equation}\label{dist-const}
	d_{\xx}(x,y)+d_{\yy}(x,y) =D.
\end{equation}
This is portrayed in Figure \ref{Fig:directed}. \\

\begin{figure}[h!]
	\includegraphics[width=5cm,height=5cm]{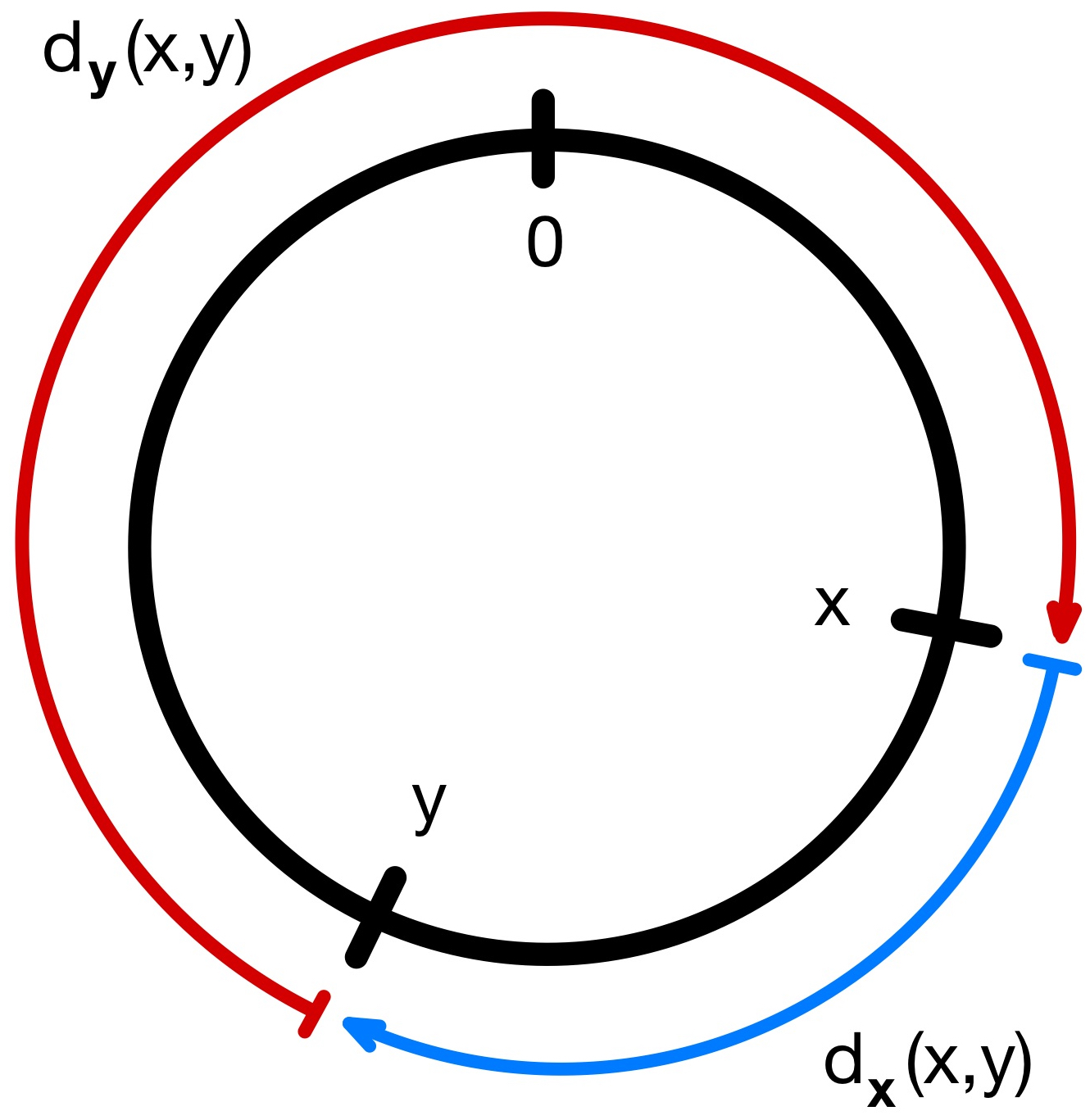}
	\caption[Directed distances]{Directed distances}\label{Fig:directed}
\end{figure}

Let us assume that players $\xx$ and $\yy$ have starting positions $x_0$ and $y_0$ in $[0,D)$. The \textit{initial minimal distance} is defined to be 
\begin{equation}\label{initial distance}
	d_0:= \min \{ d_{\xx} (x_0,y_0)  , d_{\yy} (x_0,y_0) \}.
\end{equation}

Now suppose that starting from $x_0$ and $y_0$, the operators drive at the respective speeds $v_{\xx}$ and $v_{\yy}$, with $v_{\xx},v_{\yy}$ in $\Gamma$, for $T$ units of time. Their final relative positions are then
\begin{equation*}\label{finalposition}
	x_T= (x_0 +  T v_{\xx} )_{mod \, D} \qquad \text{and}  \qquad y_T =  (y_0 +  T v_{\yy} )_{mod \, D}.
\end{equation*}

We orient the maximum displacement of buses by requiring $T v_{max}$, with $v_{max}$ given in \eqref{speedbound}, to be small compared to $D$. The reason for this is to be consistent with our assumption of constant speed strategies, since they are short-term. More precisely, we require
\begin{equation}\label{assump:short}
	T v_{max}   < \frac{D}{2}.
\end{equation}

Lastly, we define the \textit{escape distance} by
\begin{equation}\label{escapedist}
	d:=T(v_{max}-v_{min}). 
\end{equation}
This gives a threshold such that if the distance between the players is shorter than $d$, then the buses can catch up to each other, given the appropriate pair of speeds. If the distance is greater than $d$, this cannot occur.

We now proceed to define the expected utility of players given the type of game being played, namely, whether it is cooperative or non-cooperative.

\begin{itemize}
	\item \textbf{Non-cooperative game}
	
	We define the utility of $\xx$ given the initial positions of players $x_0$ and $y_0$, and the strategies $v_{\xx}$ and $v_{\yy}$, to be
	\begin{equation}\label{2-noncoop-utilx}
		u_{\xx}(x_0,v_{\xx},y_0,v_{\yy}) := \left\{ \begin{array}{cc}
			p \lambda \, d_{\xx}(x_T,y_T) - c T  &  \textrm{ if } x_T \neq y_T, \\
			p \lambda \, \frac{D}{2}- c T  & \textrm{ if } x_T=y_T .
		\end{array} \right. 
	\end{equation}
	
	The definition above includes two summands: the first one gives the (gross) expected income of $\xx$, since the factor $ p \lambda $ is the expected income per unit of distance. The second term gives the total driving cost. 
	
	It is worth pointing out that for simplicity, we have assumed that the expected income depends only on the relative final positions $x_T$ and $y_T$. A more precise account would consider the entire trajectory of the buses. Nevertheless, even if this could be described with mathematical precision, the model would grow greatly in complexity without adding to its economic interpretation. 
	
	Similarly, we define
	\begin{equation}\label{2-noncoop-utily}
		u_{\yy}(x_0,v_{\xx},y_0,v_{\yy}) := \left\{ \begin{array}{ll}
			p \lambda \, d_Y(x_T,y_T) - c T  &  \textrm{ if } x_T \neq y_T, \\
			p \lambda \, \frac{D}{2}- c T  & \textrm{ if } x_T=y_T .
		\end{array} \right. 
	\end{equation}
	
	By equation \eqref{dist-const} and the definition of the utility functions \eqref{2-noncoop-utilx}, \eqref{2-noncoop-utily}, the sum $u_{\xx} + u_{\yy}$ is a constant that does not depend on the driving speeds nor on the initial positions. For this reason, we analyze the game as a zero-sum game.
	
	\item \textbf{Cooperative game}
	
	Players aim to maximize the collective payoff, and this amounts to solving the global optimization of the sum $U_{\xx} +U_{\yy}$, which includes the utility functions in the non-cooperative game \eqref{2-noncoop-utilx} and \eqref{2-noncoop-utily}. Since the non-cooperative game is a zero-sum game, we introduce an extra term in the utility, which gives the discomfort players derive from payoff inequality. This assumption can be imagined in a situation where equity in payments is desirable, specially since players have complete information. 
	
	We define the utility function to be
	\begin{equation}\label{coop-util}
		u(x_0,v_{\xx},y_0,v_{\yy}) := u_{\xx}(x_0,v_{\xx},y_0,v_{\yy}) + u_{\yy}(x_0,v_{\xx},y_0,v_{\yy}) - k | u_{\xx}(x_0,v_{\xx},y_0,v_{\yy}) - u_{\yy}(x_0,v_{\xx},y_0,v_{\yy}) |,
	\end{equation}
	where $k$ is a non-negative constant, and all the other elements are the same as in the non-cooperative game.	
\end{itemize}~\\

\subsubsection{Mixed strategies and $\varepsilon$-equilibria}~\\

For the solution of two player games, it is convenient to define the expected utility of randomizing over the set of strategies. We also introduce the definition of $\varepsilon$-equilibrium.

Suppose that players $\xx$ and $\yy$ use the mixed strategies $X$ and $Y$.\footnote{Recall that a mixed strategy is a random variable taking values in the set $\Gamma=\{v\geq0:v_{min}\leq v\leq v_{max}\}$.} We define the utility of player $\xx$ to be 
$$ U_{\xx} ( x_0,X,y_0,Y)  := \mathbb{E} [ u_{\xx} ( x_0,X,y_0,Y)  ].$$
An analogous definition can be derived for player $\yy$.

Let $\varepsilon>0$. We say that a pair of pure strategies $(v^*_{\xx},v^*_{\yy})$ is an $\varepsilon$-equilibrium if for every $v_{\xx}$ and $v_{\yy}$ we have
$$  u_{\xx} (x_0,v_{\xx},y_0, v^*_{\yy} )   \leq u_{\xx} (x_0,v^*_{\xx},y_0, v^*_{\yy} ) + \varepsilon ,$$
and
$$ u_{\yy} (x_0,v^*_{\xx},y_0, v_{\yy} )   \leq u_{\yy} (x_0,v^*_{\xx},y_0, v^*_{\yy} ) + \varepsilon .$$

This means that any unilateral deviation from the equilibrium strategy leads to a gain of no more than $\varepsilon$; this is why an $\varepsilon$-equilibrium is also called near-Nash equilibrium. Note that in particular, an $\varepsilon$-equilibrium with $\varepsilon=0$ gives the standard definition of Nash equilibrium. However, an $\varepsilon$-equilibrium for all $\varepsilon$ sufficiently small, need not be a Nash equilibrium, specially if the utility function is discontinuous, which is our case.

A mixed strategies $\varepsilon$-equilibrium $(X,Y)$ is similarly defined by replacing the utility functions with the expected utility functions in the last definition. 

\section{Results}\label{Sec:Results}

In what follows, we analyze the speeds that drivers choose, both in the short and long-run. Results on the short term are crucial to the analysis, as implementing the optimal short-term strategies over a long period of time, gives the long-term solution to the games.

\subsection{Single player games}
The single player games have pure strategy Nash equilibria. Although the results are immediate, we include them in the analysis for completeness and ease of interpretation.

\begin{prop}\label{theo:1play}
	Let $v^*$ in $\Gamma$ be the driving speed that maximizes the utility of the driver. We provide an explicit description of $v^*$.
	\begin{enumerate}
		\item[a)] Fixed-distance game. Given the utility function defined in \eqref{1-dist-uti}, we have
		
		\begin{align*}
			v^*=\begin{cases}
				v_{min}& \text{if $p\lambda>c$},\\
				v_{min}\leq v\leq v_{max}& \text{if $p\lambda=c$},\\
				v_{max}& \text{if $p\lambda<c$}.
			\end{cases}
		\end{align*}
		\item[b)] Fixed-time game. Given the utility function defined in \eqref{1-time-uti}, we have $v^*=v_{max}$.
	\end{enumerate}
\end{prop}

\begin{proof}
	 Note that in the fixed-distance game, $p \lambda - c$ gives the driver's expected net income per unit of time. If this amount is positive, then the player maximizes her utility by driving for the longest time, or equivalently, by driving at the lowest possible speed. Conversely, a negative expected net income leads to driving at the highest speed. Lastly, a null expected income makes the driver indifferent between any given speed in the range.
	
	In the fixed-time game, the total revenue is proportional to the traveled distance, so the driver maximizes her utility by driving at the highest speed.
\end{proof}

\subsection{Two-player games}

The strategies adopted by the players strongly depend on the initial minimal distance defined in \eqref{initial distance}. We cover all cases.

\begin{thm}\label{theo:2noncoop} \textbf{Non-cooperative game.} 
	Without loss of generality we can assume $d_0=d_{\xx}(x_0,y_0)$. 
	\begin{enumerate}
		\item[a)] If $d_0=0$, that is, if the initial positions of the players are the same, then the pair of strategies $(v_{max},v_{max})$ is the only Nash equilibrium.
		\item[b)] If $0<  d_0< d < d_{\yy}(x_0,y_0)$, with $d$ the escape distance in \ref{escapedist}, then for sufficiently small $\varepsilon$, the mixed strategy $\varepsilon$-equilibria $(X,Y)$ is
		\[
		\begin{minipage}{.35\linewidth}
		\centering
		\begin{align*}
		X=\begin{cases}
		v_{min} & \text{with probability } 1-\frac{d-d_0}{D},\\
		U & \text{with probability } \frac{d-d_0}{D}
		\end{cases}
		\end{align*}
		\end{minipage}
		\quad \text{and} \quad
		\begin{minipage}{.35\linewidth}
		\centering
		\begin{align*}
		Y=\begin{cases}
		v_{min}& \text{with probability } q_1,\\
		V & \text{with probability } q_2,\\
		v_{max} -\frac{d_0}{T}+\frac{\varepsilon}{T} & \text{with probability } 1 - \frac{d}{D},
		\end{cases}
		\end{align*}
		\end{minipage}
		\]
		where $U$ is a uniform random variable on $\Big( v_{min}+ \frac{d_0}{T}, v_{max} \Big)$, $q_1$ and $q_2$ are non-negative numbers such that $q_1+q_2= \frac{d}{D}$ and $q_2 \leq \frac{d-d_0}{D}$, and $V$ is a uniform random variable on 
		$\Big( v_{max} - \frac{d_0}{T} - q_2 \frac{D}{T}, \, v_{max} - \frac{d_0}{T} \Big)$.
		
		In other words, $X$ has an atom at $v_{min}$, and $Y$ has two atoms at $v_{min}$ and $v_{max} -\frac{d_0}{T}+\frac{\varepsilon}{T}$, and are otherwise uniformly distributed over their respective intervals.
		
		\item[c)] If  $0< d= d_0< d_{\yy}(x_0, y_0)$, then for sufficiently small $\varepsilon$, the mixed strategy $\varepsilon$-equilibria is
		\[
		\begin{minipage}{.3\linewidth}
		\centering
		\begin{align*}
		X=\begin{cases}
		v_{min}& \text{with probability } 1-\frac{2 \varepsilon}{D},\\
		v_{max} & \text{with probability } \frac{2 \varepsilon}{D}
		\end{cases}
		\end{align*}
		\end{minipage}
		\quad \text{and} \quad
		\begin{minipage}{.3\linewidth}
		\centering
		\begin{align*}
		Y=\begin{cases}
		v_{min}& \text{with probability } \frac{2 d}{D},\\
		v_{min} + \frac{\varepsilon}{T}& \text{with probability } 1- \frac{2 d}{D}.
		\end{cases}
		\end{align*}
		\end{minipage}
		\]
	
		\item[d)] If  $d < d_0$, then the pair of strategies $(v_{min},v_{min})$ is the unique Nash equilibrium.
	\end{enumerate}
\end{thm}
\begin{proof}
	The proof is in Appendix \ref{app}.
\end{proof}

By assumption \eqref{assump:short}, this result covers all the possible initial positions $(x_0,y_0)$, so we have a complete and explicit characterization of the equilibria. Simply put, the theorem asserts that if the players have the same starting point, they drive at the maximum speed. If their positions differ by at most the escape distance, then they play mixed strategies. Lastly, if the distance between them is greater than the escape one, they drive at the minimum speed. See Figure \ref{Fig:Theo2} for an illustration of the result and its cases.\\

\begin{figure}[h!]
	\includegraphics[width=17cm,height=11cm]{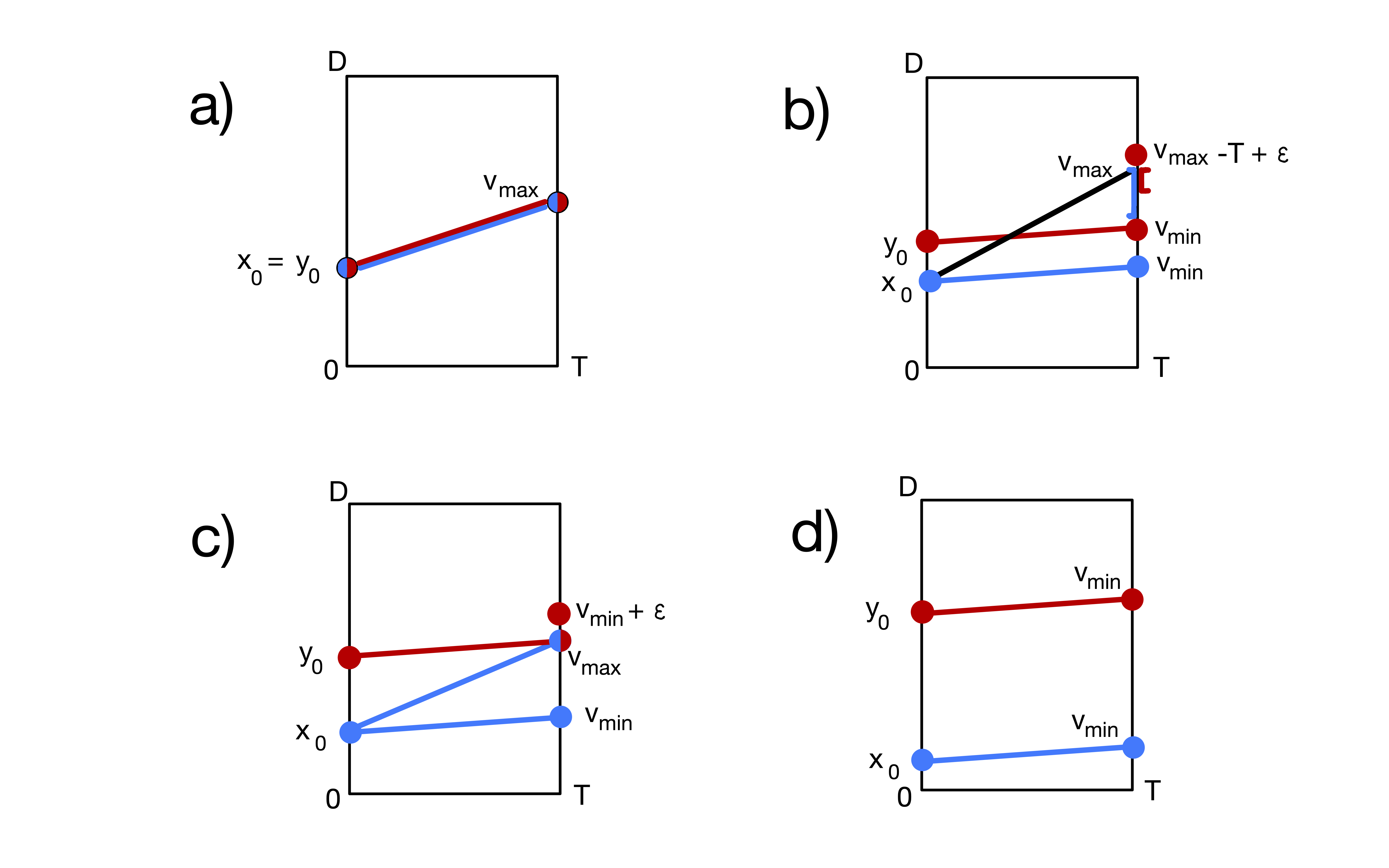}
	\caption[Time-space depictions of Theorem 2. On the rightmost side of each graph are the final positions of players, blue for $\xx$ and red for $\yy$, after driving at the optimal speed for $T$ units of time. Points represent probability mass atoms, while continuous bars give the intervals in which the locations may be.]{On the rightmost side of each graph are the final positions of players, blue for $\xx$ and red for $\yy$, after driving at the optimal speed for $T$ units of time. Points represent probability mass atoms, while continuous bars give the intervals in which the locations may be.}\label{Fig:Theo2}
\end{figure}

\begin{thm}\label{theo:coop} \textbf{Cooperative game.}
	Without loss of generality we assume that $d_0=d_{\xx}(x_0,y_0)$.
	
	\begin{enumerate}
		\item[a)]  If $d_0=0$, then the optimal pairs of driving speeds are $(v_{min}, v_{max})$ and $(v_{max}, v_{min})$.
		\item[b)]  If $0 < d_0$ and $d_0 + d < \frac{D}{2}$, then the only optimal strategies are $(v_{min}, v_{max})$.
		\item[c)]  If $d_0 + d > \frac{D}{2}$, then any pair $(v_{\xx}, v_{\yy})$ such that $T(v_{\yy}- v_{\xx} )= \frac{D}{2}$ is an optimal strategy.
	\end{enumerate}	
\end{thm}

\begin{proof}
	The proof is direct. Since the sum $u_{\xx}(x_0,v_{\xx},y_0,v_{\yy}) + u_{\yy}(x_0,v_{\xx},y_0,v_{\yy})$ is equal to a constant for any pair $(v_{\xx},v_{\yy})$, the only quantity left to optimize is $- k | u_{\xx}(x_0,v_{\xx},y_0,v_{\yy}) - u_{\yy}(x_0,v_{\xx},y_0,v_{\yy}) | $. Minimization occurs when the distance between the final positions $x_{F}$ and $y_{F}$ is the greatest possible. It is easy to check that the driving speeds listed above do just this. 
\end{proof}

An important observation is that in the case where $d_0=\frac{D}{2}$, which is accounted for in c) all the optimal strategies are of the form $(v,v)$ for a feasible speed $v$. Intuitively, this means that if the players have diametrically opposite initial positions, then any speed is optimal, as long as both adopt it.

\subsection{Long-run analysis}\label{subsect:time}

Let us recall that the previous results are obtained for small enough $T$, the formal requirement being stated in \eqref{assump:short}. It is of interest to know what happens in longer time periods, and in particular, in the long-run. To this end, we repeat the games infinitely many times, implementing the optimal strategies in each stage. Of course, the strategies depend on the distance between players, which is given by the implementation of the optimal strategies in the previous period. It is thus convenient to define a recursive process, and to introduce a few variables.

Consider the initial positions of $\xx$ and $\yy$, namely $(x_0,y_0)$, with $d_0$ defined in \eqref{initial distance}. Let $\{(x_n,y_n) \}_{n \geq 1}$ be a stochastic process with the following property: the pair $(x_{k+1}, y_{k+1})$ gives the final locations of the players after they play their optimal strategies, taking $(x_k,y_k)$ as their starting positions. It is worth noting that since equilibria in Theorem \ref{Fig:Theo2} involve mixed strategies, randomness is very much present in the process.

We define the distance between the buses at any (non-negative integer) time as:
\begin{equation}\label{dist_evo}
	d_n:= \min \{ d_{\xx}(x_n,y_n), d_{\yy}(x_n,y_n) \} \qquad \forall \enspace n \geq 0.
\end{equation}
We also define the first time in which $d_n$ exceeds the escape distance $d$ (given in \eqref{escapedist}), denoted by $N$, as follows
$$ N = \min \{ n \geq 0: d_n > d \} .$$

\begin{thm}\label{theo:repeat}\textbf{Non-cooperative game.}
	If $d_0 \neq 0, d$, we have
	$$ \mathbb{P}  ( N > k)  \leq \Big(\frac{d}{D} \Big) ^{k}   \quad \text{for all $k \geq 1$}.$$
	If $d_0 = d$, then there exists a geometrically distributed random time $M$ with parameter $1- (1- \frac{2\varepsilon}{D}  )(\frac{2d}{D})$, taking values in the natural numbers, with $\varepsilon$ satisfying the $\varepsilon$-equilibrium conditions in Theorem \ref{Fig:Theo2}, with the property that $d_k=d$ for all $k<M$, and
	
	\begin{displaymath}
		d_M= \left\{ \begin{array}{ll}
			0 & \textrm{ with probability } \qquad \frac{ 4 \varepsilon d }{ D^2 \Big( 1- \Big(1- \frac{2 \varepsilon}{D} \Big) \Big( \frac{2d}{D} \Big) \Big)  },\\
			>d &  \textrm{ with complementary probability.}  \\
		\end{array} \right. 
	\end{displaymath}
\end{thm}

\begin{proof}
	For the proof we refer the reader to Appendix \ref{app}.
\end{proof}

Explicitly, this means that for most starting points, playing the game repeatedly leads to a bus gap greater than the escape distance in a finite and geometrically distributed time. From Theorem \ref{Fig:Theo2}, we conclude that in this case, drivers end up driving at the minimum speed. There are two exceptions to this: if the drivers have the same starting position, or if the initial distance between them is exactly that of escape. In the former case, the drivers choose to go at the maximum speed forever, and in the latter, they maintain their distance for some random time, and from then on reach the escape distance, and drive at the minimum speed. It is with very little probability (proportional to $\varepsilon$) that this scenario does not occur. Figure
\ref{Fig:Theo3} shows the evolution of the distance process $\{d_n: n \geq 0 \}$ given a few initial distances $d_0$. 

\begin{figure}[h!]
	\includegraphics[width=10cm,height=7.7cm]{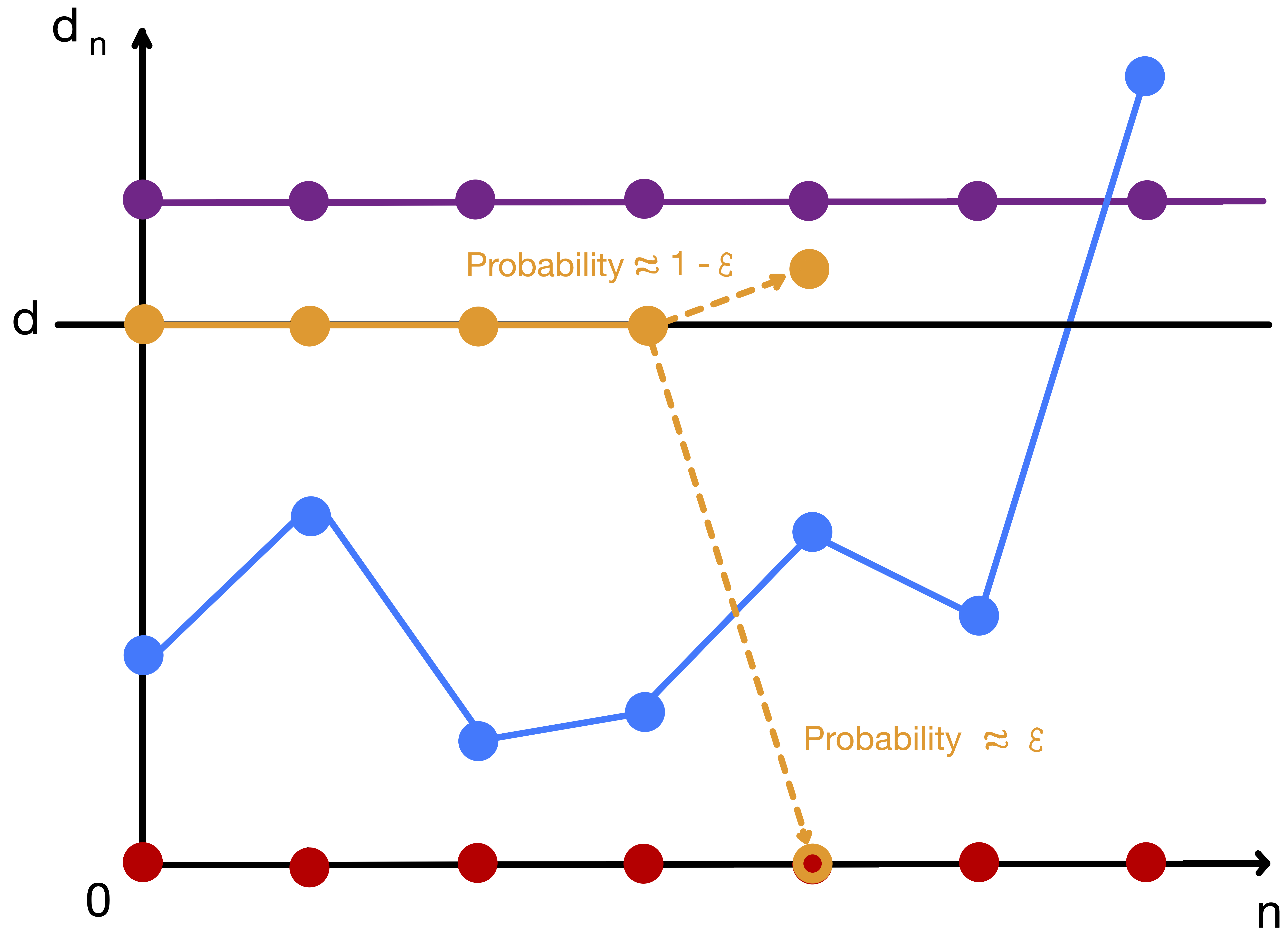}
	\caption[Evolution of the process $\{d_n: n \geq 0 \}$ for different initial positions, non-cooperative game]{Evolution of the process $\{d_n: n \geq 0 \}$ for different initial positions, non-cooperative game.}\label{Fig:Theo3}
\end{figure}

\begin{thm}\label{theo:repeatcoop}\textbf{Cooperative game.}
	For all $d_0 \geq 0$ we have $N \leq \lceil \frac{  D }{2 d} \rceil$, where $N= \min \Big\{ n \geq 0 :  d_n = \frac{D}{2} \Big\}$, and $\lceil z \rceil$ is the least integer greater than or equal to the real number $z$.
\end{thm}

\begin{proof}
	First note that $N$ gives the time in which the buses reach diametrically opposite positions in the circuit. Also, playing the optimal strategies in Theorem \ref{theo:coop}, increases the distance between the buses by $d$. Hence, repeating the game eventually leads to reaching the diametric distance. This means that $N$ is at most the number of steps of size $d$ necessary to go over $\frac{D}{2}$. Once diametrical positions are reached, the distance is preserved forever.
\end{proof}

\subsection{Extension}\label{sub:ext}
It is possible to account for perturbations like traffic lights, congestion, or accidents, by introducing a random noise to the displacement of buses. One could do this defining
\begin{equation}\label{noisemodel}
	x_T= (x_0 +  T v_{\xx} + \sigma Z_x )_{mod \, D}  \quad\text{and}\quad y_T =  (y_0 +  T v_{\yy} + \sigma Z_y )_{mod \, D}, 
\end{equation}
where $Z_x$ and $Z_y$ are independent standard normal random variables and $\sigma \geq 0$ is a fixed parameter. 

Then, the following results would be observed.

\begin{itemize}
	\item \textbf{Non-cooperative game.} Given that the expected value of the final positions is unchanged, Theorem \ref{theo:2noncoop} remains valid. However, the repetition of this new game leads to a new result. Since the probability of maintaining a null, or escape distance $d$, at any positive time is zero, the long-run analysis is reduced to two distinct cases: $0 < d_0  < d$ and $d_0 < d$. Arguments similar to that in the proof of Theorem \ref{theo:repeat} show that if $0< d_0 < d$, we have $d_N \geq d$ in an exponentially fast time $N$. If $d < d_0$, then the distance process $\{d_n\}_{n \geq 1}$ remains above $d$ for a random time $M$, but eventually falls below it. The expected time above is inversely proportional to $\sigma$.
	
	\item \textbf{Cooperative game.}  The analysis collapses to the cases b) and c) of Theorem \ref{theo:coop}. So, while the players try to reach the diametrically opposite positions, with probability one this does not occur.
\end{itemize}

\section{Concluding remarks}\label{Sec:Conclu}

Our theoretical results are consistent with the driving practices mentioned in the Introduction. In particular, Theorem \ref{theo:2noncoop}.$a$ induces (2) Racing, Theorem \ref{theo:2noncoop}.$b$, $c$ conduce to (2) Racing and (3) Overtaking, Tailing or Chasing, and Theorem \ref{theo:2noncoop}.$d$ to (1) Hanging back or Crowling. It is worth noting that all of the aforementioned are short-term strategies. As far as the time-evolution of the game goes, Theorem \ref{theo:repeat} asserts that in the long run and with high probability, both operators end up hanging back. Theorems \ref{theo:coop} and \ref{theo:repeatcoop} are intended to contrast the drivers' optimal strategies and ultimately the equilibria when cooperation is desired.

In subsection \ref{sub:ext}, we extended the model to allow for randomness in displacement. In this scenario no equilibrium is lasting, so the operators alternate between racing, hanging back and chasing from time to time. We believe this is precisely what happens in Mexico City, although proving this would require a data driven approach analysis. 

There are a few open problems worth exploring. First, one could increase the number of players, and investigate whether equilibria still exists, and if so, try to characterize it. Second, one may vary the distribution of the passengers along the route, dispensing with the homogeneous assumption. Along these lines, one may introduce traffic congestion by making the utility function depend on space in a non-homogeneous manner. This would potentially require strategies to depend on the player's position. Lastly, one could introduce decision variables like tariffs and timetables; doing so would allow to compare the results with some that have already been addressed in the literature.

\subsection*{Declaration of interest}
None.

\bibliographystyle{abbrv}
\bibliography{Buses.bib}

\begin{thebibliography}{10}

\bibitem{Acta}
{Act of Parliament, United Kingdom}.
\newblock Transport {A}cta.
\newblock 1985.

\bibitem{Adler}
N.~Adler, A.~Brudner, and S.~Proost.
\newblock A review of transport market modeling using game-theoretic
  principles.
\newblock {\em European Journal of Operational Research}, 2020.

\bibitem{Baik}
J.~Baik, A.~Borodin, P.~Deift, and T.~Suidan.
\newblock A model for the bus system in {C}uernavaca ({M}exico).
\newblock {\em Journal of Physics A: Mathematical and General},
  39(28):8965--8975, 2006.

\bibitem{Borenstein}
S.~Borenstein and J.~Netz.
\newblock Why do all the flights leave at 8 am?: Competition and departure-time
  differentiation in airline markets.
\newblock {\em International Journal of Industrial Organization}, 17(5):611 --
  640, 1999.

\bibitem{EV2019}
{Consultant Buend\'ia y Laredo}.
\newblock Encuesta sobre victimizaci\'on en el transporte p\'ublico en la
  {C}iudad de {M}\'exico y en la {Z}ona {M}etropolitana 2019.
\newblock 2019.

\bibitem{dePalma}
A.~{de Palma} and R.~Lindsey.
\newblock Optimal timetables for public transportation.
\newblock {\em Transportation Research Part B: Methodological}, 35(8):789 --
  813, 2001.

\bibitem{Dodgson}
J.~S. Dodgson, Y.~Katsoulacos, and C.~R. Newton.
\newblock An {A}pplication of the {E}conomic {M}odelling {A}pproach to the
  {I}nvestigation of {P}redation.
\newblock {\em Journal of Transport Economics and Policy}, 27(2):153--170,
  1993.

\bibitem{Ellis}
C.~J. Ellis and E.~C. Silva.
\newblock British {B}us {D}eregulation: Competition and {D}emand
  {C}oordination.
\newblock {\em Journal of Urban Economics}, 43(3):336 -- 361, 1998.

\bibitem{Evans}
A.~Evans.
\newblock A {T}heoretical {C}omparison of {C}ompetition with {O}ther {E}conomic
  {R}egimes for {B}us {S}ervices.
\newblock {\em Journal of Transport Economics and Policy}, 21(1):7--36, 1987.

\bibitem{Foster}
C.~Foster and J.~Golay.
\newblock Some {C}urious {O}ld {P}ractices and {T}heir {R}elevance to
  {E}quilibrium in {B}us {C}ompetition.
\newblock {\em Journal of Transport Economics and Policy}, 20(2):191--216,
  1986.

\bibitem{Gelf}
A.~E. Gelfand, P.~J. Diggle, M.~Fuentes, and P.~Guttorp.
\newblock Handbook of spatial statistics.
\newblock {\em Statistics in Medicine}, 30(8):899--900, 2011.

\bibitem{Comparative}
{Global Green Growth Institute}.
\newblock Comparative {A}nalysis of {B}us {P}ublic {T}ransport {C}oncession
  {M}odels.
\newblock 2018.

\bibitem{Hot}
H.~Hotelling.
\newblock Stability in {C}ompetition.
\newblock {\em The Economic Journal}, 39(153):41--57, 1929.

\bibitem{Ibarra}
O.~Ibarra-Rojas, F.~Delgado, R.~Giesen, and J.~Munoz.
\newblock Planning, operation, and control of bus transport systems: A
  literature review.
\newblock {\em Transportation Research Part B: Methodological}, 77:38 -- 75,
  2015.

\bibitem{EOD2017}
{Instituto Nacional de Estad\'istica y Geograf\'ia}.
\newblock Encuesta {O}rigen {D}estino en {H}ogares de la {Z}ona {M}etropolitana
  del {V}alle de {M\'e}xico 2017.
\newblock 2017.

\bibitem{Ireland}
N.~J. Ireland.
\newblock A {P}roduct {D}ifferentiation {M}odel of {B}us {D}eregulation.
\newblock {\em Journal of Transport Economics and Policy}, 25(2):153--162,
  1991.

\bibitem{Newell}
G.~F. Newell.
\newblock Dispatching {P}olicies for a {T}ransportation {R}oute.
\newblock {\em Transportation Science}, 5(1):91--105, 1971.

\bibitem{Salop}
S.~C. Salop.
\newblock Monopolistic {C}ompetition with {O}utside {G}oods.
\newblock {\em The Bell Journal of Economics}, 10(1):141--156, 1979.

\bibitem{PEM2019}
{Secretar\'ia de Movilidad, Gobierno de la Ciudad de M\'exico}.
\newblock Plan {E}strat\'egico de {M}ovilidad 2019. {U}na ciudad, un sistema.
\newblock 2019.

\end{thebibliography}

\appendix{}
\section{Computations}\label{app}

To prove Theorem \ref{theo:2noncoop}, it is convenient to introduce the following Lemma.

\begin{lem}\label{lemma:mixed}
	Let $X$ be a mixed strategy of $\xx$ and $Y$ be a mixed strategy of $\yy$. We define $Z$ to be a mixed random variable in the Probability theory sense: it has both discrete and continuous components. In particular, $Z$ is of the form
	\begin{eqnarray*}
		Z= \left\{ \begin{array}{cc}
			z_i &  \textrm{ with probability } p_i, \textrm{ for } i \in I ,\\
			W & \textrm{ with probability }  1-  \sum_{i \in I} p_i,
		\end{array} \right. 
	\end{eqnarray*}
	where $I$ is a finite or numerable set, and $W$ is a continuous random variable with density $f_W(t)$ on its support, denoted by $supp(f_W)$. Then,
	\begin{eqnarray}\label{lemma:mixed1}
		U_\xx ( x_0, X, y_0, Y) &=& \sum_{i \in I} \mathbb{E} (  u_\xx ( x_0, X, y_0, Y) | Z=z_i) \, p_i  \nonumber \\
		&+&  \Big( 1-  \sum_{i \in I} p_i \Big) \int_{supp(f_W)}  \mathbb{E} (  u_\xx ( x_0, X, y_0, Y) | Z=w) f_W(w) dw .
	\end{eqnarray}
	If $Z=X$ and $(X,Y)$ is a mixed strategy Nash equilibrium, then
	\begin{equation}\label{lemma:mixed2}
		\mathbb{E} (  u_\xx ( x_0, X, y_0, Y) | X=z_i) =  \int_{supp(f_W)}  \mathbb{E} (  u_\xx ( x_0, X, y_0, Y) | X=w) f_W(w) dw  \qquad \forall i \in I,
	\end{equation}
	and
	\begin{equation}\label{lemma:mixed3}
		\mathbb{E} (  u_\xx ( x_0, X, y_0, Y) | X=w_1)=  \mathbb{E} (  u_\xx ( x_0, X, y_0, Y) | X=w_2) \qquad \forall w_1,w_2 \in supp(f_W).
	\end{equation}
\end{lem}

\begin{proof}
	Equation \eqref{lemma:mixed1} is straightforwardly obtained by computing the conditional expectancy of the random variable $u_\xx ( x_0, X, y_0, Y)$
	given the values of $Z$.\\
	
	Note that if \eqref{lemma:mixed2} does not occur, then there exist two different values $z_i$ and $z_j$, such that $\mathbb{E} ( u_\xx ( x_0, X, y_0, Y) | X=z_i)\neq\mathbb{E} ( u_\xx ( x_0, X, y_0, Y) | X=z_j)$. This means that $U_{\xx}$ can be increased by placing all the probability on the value that gives the highest expectation. This leads to a contradiction with the form of the mixed strategy $X$. Similar arguments apply to the case where \eqref{lemma:mixed2} is violated through the continuous component. 
	
	Likewise, if condition \eqref{lemma:mixed3} is not fulfilled, then there are two values $w_1$ and $w_2$ such that  \linebreak
	$\mathbb{E} (  u_\xx ( x_0, X, y_0, Y) | X=w_i)$ are different. Then, $U_{\xx}$ can be increased by restricting the support of $f_W$ to the points where the maximum of the function $g(w)=\mathbb{E} (  u_\xx ( x_0, X, y_0, Y) | X=w)$ is reached. Here, the form of the mixed strategy $X$ is violated.
\end{proof}

\textbf{Proof of Theorem \ref{theo:2noncoop}}: \\

First, note that for optimizing the utility function \eqref{2-noncoop-utilx}, \eqref{2-noncoop-utily} the terms $p \lambda$ and $c$ are irrelevant, since the $\argmin$ of any function is invariant under linear transformations. Thus, there is no loss of generality in assuming that $p \lambda =1$ and $c=0$.

By equation \eqref{loopinv}, we may actually assume that $0=x_0\leq y_0 < D$. We then have
$$  d_0= d_{\xx} (x_0,y_0)  = y_0 \quad\text{and}\quad  d_{\yy} (x_0,y_0) = D-y_0 .$$
Under the above assumption and using \eqref{speedbound}, \eqref{assump:short} in cases a), b), c) and d), it happens that $ 0 < x_T, y_T < D$, so we can get rid of all the $D$-modules in the computations. 

For computing the $\varepsilon$-equilibrium, we will consider the $\varepsilon$-best reply, defined as follows.
Let $\varepsilon$ be a positive number. We say that a strategy $v^*_{\xx}$ is $\xx$'s $\varepsilon$-best reply to $\yy$'s strategy $v_{\yy}$, if 
$$  u_{\xx} (x_0,v_{\xx},y_0, v_{\yy} )   \leq u_{\xx} (x_0,v^*_{\xx},y_0, v_{\yy} ) + \varepsilon ,$$
for all strategies $v_{\xx}$.

To simplify notation, we write $u_{\xx} (v_{\xx},v_{\yy} )$ and $u_{\xx} (X,Y)$ in the case of mixed strategies, instead of $u_{\xx} (x_0,v_{\xx},y_0, v_{\yy} )$ and $u_{\xx}(x_0,X,y_0,Y)$ if the computations do not depend on the fixed initial positions.\\ 

\begin{itemize}
	\item \textbf{Case a)}\ \\
	We assume that $x_0=y_0=0$. Let player $\yy$ pick the strategy $v_{\yy}= v_{max}$. Then,
	$$ u_{\xx}(x_0,v_{\xx},y_0,v_{min}) =  d_{\xx}( T v_{\xx}, T v_{max}) = T v_{max} - T v_{\xx}   \leq  T v_{max} .$$
	Using \eqref{assump:short}, we obtain the bound
	$$  u_{\xx}(x_0,v_{\xx},y_0,v_{min}) \leq  \frac{D}{2}  = u_{\xx}(x_0,v_{max},y_0, v_{max} ) .$$
	Explicitly, this means that the strategy $v_{\xx}=v_{max}$ is the best reply to $v_{\yy}=v_{max}$. By symmetry, we conclude that $(v_{max}, v_{max})$ is a Nash equilibrium.
	
	To check the uniqueness of the equilibrium, we note that $\yy$'s $\varepsilon$-best reply to a given speed $v_{\xx} < v_{max}$ chosen by $\xx$, is $v_{\yy}=  v_{\xx} + \epsilon$ for sufficiently small $\varepsilon$. On the other hand, $\xx$'s  $\varepsilon$-best reply to $v_{\yy}=  v_{\xx} + \epsilon$ is $v_{\xx}=v_{\yy} + \varepsilon$. Therefore the only equilibrium is $(v_{max}, v_{max})$.
	
	\item \textbf{Case b)}\ \\
	Let us denote by $B_{\xx}(v)$ $\xx$'s best reply when $\yy$ plays $v$. It is straightforward to show that 
	\begin{eqnarray*}
		B_{\xx}(v)= \left\{ \begin{array}{ll}
			v + \frac{d_0}{T} + \frac{\varepsilon}{T} &  \textrm{ if } v_{min} \leq v < v_{max} - \frac{d_0}{T} , \\
			v_{max} & \textrm{ if } v = v_{max} - \frac{d_0}{T} , \\
			v_{min}  & \textrm{ if }  v_{max} - \frac{d_0}{T} < v,
		\end{array} \right. 
	\end{eqnarray*}
	and \begin{eqnarray*}
		B_{\yy}(v)= \left\{ \begin{array}{ll}
			v_{min} & \textrm{ if } v  < v_{min} + \frac{d_0}{T} , \\
			v - \frac{d_0}{T} + \frac{\varepsilon}{T}  & \textrm{ if }  v_{min} + \frac{d_0}{T} \leq v \leq v_{max},
		\end{array} \right. 
	\end{eqnarray*}
	under hypothesis $b)$.
	
	If $(X,Y)$ is a mixed strategy Nash equilibrium, then the support of the random variable $X$ should be contained in the set of $\xx$'s best replies, the corresponding is true for variable $Y$. In this particular case, $X$ has support on $\{ v_{min } \} \cup ( v_{min} + \frac{d_0}{T}, v_{max} ) \cup \{ v_{max} \}$, while $Y$ has support on $\{ v_{min } \} \cup (v_{min} , v_{max} - \frac{d_0}{T} ) \cup \{ v_{max} - \frac{d_0}{T} + \frac{\varepsilon}{T} \}$.
	
	Hence, a mixed strategy $X$ with the support obtained is of the form
	\begin{eqnarray*}
		X= \left\{ \begin{array}{ll}
			v_{min} &  \textrm{ with probability } p_1 , \\
			U & \textrm{ with probability } p_2 , \\
			v_{max}  & \textrm{ with probability } 1-p_1-p_2,
		\end{array} \right. 
	\end{eqnarray*}
	where $p_1,p_2 \in [0,1]$, and $U$ is a continuous random variable with density $f_U(u)$ and support contained in $(v_{min} + \frac{d_0}{T}, v_{max})$.
	Similarly, a mixed strategy $Y$ with the desired support is 
	\begin{eqnarray*}
		Y= \left\{ \begin{array}{ll}
			v_{min} &  \textrm{ with probability } q_1 , \\
			V & \textrm{ with probability } q_2 , \\
			v_{max} - \frac{d}{T} + \frac{\varepsilon}{T}  & \textrm{ with probability } 1-q_1-q_2,
		\end{array} \right. 
	\end{eqnarray*}
	where $q_1,q_2 \in [0,1]$, and $V$ is a continuous random variable with density $f_V(v)$ with support contained in $(v_{min}, v_{max} - \frac{d_0}{T})$. 
	
	To compute the density of $U$, we apply \eqref{lemma:mixed3} to $Y$. Let us compute $\esp ( u_{\yy} (X,Y) | Y=v)$ when $v \in (v_{min} - \frac{d_0}{T}, v_{max} - \frac{d_0}{T} )$: 
	
	\begin{eqnarray*}
		\esp ( u_{\yy} (X,Y) | Y=v) &=&  \esp ( u_{\yy} (X,v)) \\
		&=& p_1 u_{\yy} ( v_{min} ,v ) + p_2 \esp ( u_{\yy} (U,v)) + (1-p_1-p_2) u_{\yy} ( v_{max} ,v)  \\
		&=& p_1 (D+ Tv_{min} -Tv-d_0 )  \\
		&+& p_2 \Big[  \int_{v_{min}}^{v+ \frac{d_0}{T}} (D+Tu-Tv-d_0) f_U(u) du + \int_{v+ \frac{d_0}{T}}^{v_{max}}
		(Tu-Tv-d_0) f_U(u) du \Big]  \\
		&+& (1-p_1-p_2) (Tv_{max} -Tv -d_0) \\
		&=& p_1 D + p_2 D F_U \Big( v +  \frac{d_0}{T} \Big) +  p_1 T v_{min} + (1-p_1-p_2) Tv_{max} + p_2 T \esp(U) -Tv -d_0,
	\end{eqnarray*}
	where $F_U(u)$ is the cumulative probability distribution function of the random variable $U$.
	
	By \eqref{lemma:mixed3}, we have
	\begin{equation}\label{k}
		p_1 D + p_2 D F_U \Big( v +  \frac{d_0}{T} \Big) +  p_1 T v_{min} + (1-p_1-p_2) Tv_{max} + p_2 T \esp(U) -Tv -d_0 =k,
	\end{equation}
	for some constant $k$.
	
	Since $F_U(v_{max}) =1$, when we plug $v= v_{max} - \frac{d_0}{T}$, we obtain its value
	
	\begin{equation}\label{valuek}
		k= (p_1+p_2)D + p_1 Tv_{min} - (p_1+p_2) Tv_{max} + p_2 T \esp(U).
	\end{equation}
	On substituting $k$ into \eqref{k} we obtain 
	$$ F_U \Big(v + \frac{d_0}{T} \Big) = 1  - \frac{T(v_{max}-v) -d_0}{p_2 D} .$$ 
	Let $u= v+\frac{d_0}{T}$. Then, $u \in (v_{min} +\frac{d_0}{T} , v_{max} )$ and $F_U(u)= 1 - \frac{T(v_{max}-u)}{p_2D}$. From this we have $u^*=v_{max}-\frac{p_2D}{T}$ is the value such that $F_U(u^*)=0$.
	
	The conclusion is that $U$ is uniformly distributed on the interval $(v_{max}- \frac{p_2D}{T} , v_{max})$, thus
	\begin{equation}\label{esperanzaU}
		\esp(U)=v_{max} - \frac{p_2D}{2T}.
	\end{equation}
	
	In the same manner we can see that $V$ has uniform distribution on the interval $(v_{max} - \frac{d_0}{T} -\frac{q_2 }{T}, v_{max} - \frac{d_0}{T} )$, with expectancy given by
	\begin{equation}\label{esperanzaV}
		\esp(V)=v_{max} - \frac{d_0}{T}- \frac{q_2 D}{2T}.
	\end{equation}
	
	To compute the values of $p_1$ and $p_2$ necessary for the $\varepsilon$-equilibrium, we use \eqref{lemma:mixed2}. We first compute the conditional expectancy of $u_{\yy}(X,Y)$ given $Y$,
	\begin{eqnarray*}
		\esp(u_{\yy}(X,Y) | Y=v_{min} ) &=& p_1 u_{\yy}(v_{min},v_{min}) +p_2 \esp( u_{\yy}(U,v_{min}) )+ (1-p_1-p_2) u_{\yy}(v_{max},v_{min} ) \\
		&=& p_1( D- d_0) + p_2 \Big[  \int_{v_{max} - \frac{p_2D}{T} }^{ v_{max} }  (Tu -Tv_{min} - d_0) f_U(u) \, du   \Big] \\
		& + & (1-p_1-p_2) ( Tv_{max}- Tv_{min}-d_0 ) \\
		&=& p_1( D- d_0) + p_2 (T\esp(U)  -Tv_{min} -d_0 )+ (1-p_1-p_2) ( Tv_{max}- Tv_{min}-d ).
	\end{eqnarray*}
	By \eqref{esperanzaU}, we have
	\begin{equation}\label{condv_min}
		\esp(u_{\yy}(X,Y) | Y=v_{min} ) = p_1D - d_0 +(1-p_1)(Tv_{max} -Tv_{min} ) -p_2^2 \frac{D}{2}.
	\end{equation}
	Computing $\esp (u_{\yy}(X,Y) | Y=V)$ yields
	$$
	\esp ( u_{\yy}(X,Y) | Y=V) =  \int_{ v_{max} - \frac{d_0}{T} -\frac{q_2 D}{T} }^{v_{max} - \frac{d_0}{T}} \esp( u_{\yy}(X,Y) | Y=v ) f_V(v) \, dv .
	$$
	Since we know that the integrand is constant and its value is given by equations \eqref{valuek} and \eqref{esperanzaU}, we directly obtain 
	\begin{equation}\label{condV}
		\esp ( u_{\yy}(X,Y) | Y=V) = (p_1+p_2)D -p_1 (Tv_{max} -Tv_{min}) - p_2^2 \frac{D}{2} .
	\end{equation}
	We are left with the task of determining the expected value of $u_{\yy}(X,Y)$ conditioned on the value $Y= v_{max} - \frac{d_0}{T} + \frac{\varepsilon}{T}$,
	\begin{eqnarray}\label{condvmaxmenos}
		\esp \Big(u_{\yy}(X,Y) | Y= v_{max} - \frac{d_0}{T} + \frac{\varepsilon}{T} \Big) &=& p_1 u_{\yy} \Big(v_{min},   v_{max} - \frac{d_0}{T} + \frac{\varepsilon}{T} \Big)
		+ p_2 \esp \Big( u_{\yy} \Big( U,v_{max} - \frac{d_0}{T} + \frac{\varepsilon}{T} \Big) \Big)  \nonumber \\
		&+& (1-p_1 - p_2) u_{\yy} \Big(v_{max}, v_{max} - \frac{d_0}{T} + \frac{\varepsilon}{T} \Big)  \nonumber\\
		&=& p_1 (D- T(v_{max} -v_{min}) - \varepsilon ) \nonumber \\
		&+& p_2 \int_{v_{max} - \frac{p_2 D}{T}}^{v_{max} } (D-T(v_{max}-u)-\varepsilon) f_U(u) \, du + (1-p_1-p_2) (D- \varepsilon) \nonumber \\
		&=& p_1 (D- T(v_{max} -v_{min}) - \varepsilon ) \nonumber \\
		&+& p_2  ( D- Tv_{max} + T \esp(U) -\varepsilon)  + (1-p_1-p_2) (D- \varepsilon) \nonumber \\
		&=& D -\varepsilon - p_1 T(v_{max}-v_{min}) -p_2^2 \frac{D}{2},
	\end{eqnarray}
	where we used \eqref{esperanzaU} in the last equality.
	
	Lemma \eqref{lemma:mixed2} implies that in order to have an $\varepsilon$-equilibrium, the
	expressions \eqref{condv_min}, \eqref{condV} and \eqref{condvmaxmenos} must be equal. This system of equations has the unique solution 
	$$ p_1 = 1-  \frac{T(v_{max}-v_{min})-d_0}{D}, \qquad p_2= \frac{T(v_{max}-v_{min})-d_0}{D} , \qquad 1 - p_1 -p_2=0 .$$
	
	We now apply this argument again, to obtain the expectancy of the random variable $u_{\xx} (X,Y)$ conditioned on the values of $X$, as well as the values $q_1,q_2$ necessary to have an $\varepsilon$-equilibrium. In this case, there are many solutions. Indeed, any combination $q_1,q_2$ satisfying $$ 0 \leq q_1,q_2 ,   \qquad q_1 + q_2 = \frac{T(v_{max} - v_{min} )}{D}, \qquad 1-q_1-q_2 = 1- \frac{T(v_{max} - v_{min} )}{D} ,$$
	fulfills equation \eqref{lemma:mixed2}.
	
	Given that the support of $V$ is $\Big( v_{max} - \frac{d_0}{T} - q_2 \frac{D}{T}, \, v_{max} - \frac{d_0}{T} \Big)  \subseteq \Big( v_{min}, v_{max} - \frac{d_0}{T} \Big)$, it is necessary to impose the condition $q_2 \leq \frac{d-d_0}{D}$.
	
	\item \textbf{Case c)}\ \\
	
	From the conditions stated in $c)$, it follows that 
	\begin{eqnarray*}
		B_{\xx}(v)= \left\{ \begin{array}{ll}
			v_{max} &  \textrm{ if } v=v_{min}, \\
			v_{min} & \textrm{ if } v > v_{min}.
		\end{array} \right. 
	\end{eqnarray*}
	
	Intuitively, under hypothesis $c)$, it always happens that $x_T \leq y_T$ for every pair of strategies $v_{\xx}, v_{\yy}$. Equality holds only when $v_{\xx}=v_{max}$ and $v_{\yy}=v_{min}$. 
	
	Similarly, one can check that
	\begin{eqnarray*}
		B_{\yy}(v)= \left\{ \begin{array}{ll}
			v_{min} &  \textrm{ if } v<v_{max}, \\
			v_{max} + \frac{\varepsilon}{T} & \textrm{ if } v = v_{max},
		\end{array} \right. 
	\end{eqnarray*}
	where last case is an $\varepsilon$-best reply.
	
	To find the $\varepsilon$-equilibria, we define $X$ to be a random variable such that
	$$ \mathbb{P} ( X=  v_{min} )= p,  \qquad  \mathbb{P} ( X=  v_{max} )= 1- p ,\quad \text{for some probability $p \in [0,1]$.}$$
	Similarly, we define a random variable $Y$ such that
	$$ \mathbb{P} ( Y=  v_{min} )= q,  \qquad  \mathbb{P} \Big( Y=  v_{min} + \frac{\varepsilon}{T} \Big)= 1- q ,\quad\text{for $q \in [0,1]$.}$$
	
	An $\varepsilon$-equilibrium requires $ \mathbb{E} ( u_{\xx}( v_{min}, Y )   ) = \mathbb{E} (u_{\xx} (v_{\max}, Y )  )$, which is exactly the condition \eqref{lemma:mixed2} when there is no continuous part for $X$.
	
	Since
	$$\mathbb{E} ( u_{\xx}( v_{min}, Y )   )= q u_{\xx} (v_{min}, v_{min}) + (1-q) u_{\xx} \Big( v_{min}, v_{min} + \frac{\varepsilon}{T} \Big) 
	= q d + (1-q) (d + \varepsilon),$$
	and 
	$$\mathbb{E} ( u_{\xx}( v_{max}, Y )   )= q u_{\xx} (v_{max}, v_{min}) + (1-q) u_{\xx} \Big( v_{max}, v_{min} + \frac{\varepsilon}{T} \Big) 
	= q \Big( \frac{D}{2} \Big) + (1-q) (\varepsilon), $$
	we can equalize the two equations and solve to obtain $q= \frac{2d}{D}$. Note that \eqref{speedbound} implies that $0<q <1$.

	Similarly, we should have $\mathbb{E} ( u_{\yy}( X, v_{min} )   ) = \mathbb{E} \Big( u_{\yy} \Big( X, v_{min} + \frac{\varepsilon}{T} \Big) \Big)$. The explicit formulas being
	$$ \mathbb{E} ( u_{\yy}( X, v_{min} ) ) = p u_{\yy}(v_{min}, v_{min} ) +(1-p) u_{\yy} ( v_{max}, v_{min} ) = p (D-d) + (1-p) \frac{D}{2},$$
	and 
	$$ \mathbb{E} \Big( u_{\yy} \Big(  X, v_{min} + \frac{\varepsilon}{T} \Big) \Big)= p \, u_{\yy} \Big(v_{min}, v_{min} + \frac{\varepsilon}{T} \Big) +
	(1-p) u_{\yy} \Big(v_{max}, v_{min} + \frac{\varepsilon}{T} \Big) = p (D-d - \varepsilon) + (1-p)(D- \varepsilon).$$
	
	Matching and solving the two yields $0<1-p= \frac{2 \varepsilon}{D}<1$. 
	
	\item \textbf{Case d)}\ \\
	
	Assume that player $\yy$ chooses strategy $v_{\yy}$ satisfying \eqref{speedbound}. Then
	\begin{equation}\label{proofTheo2:eq1}
		u_{\xx}(x_0,v_{\xx},y_0,v_{\yy}) = d_{\xx}( T v_{\xx}, y_0 + T v_{\yy} ).  
	\end{equation}
	
	By assumption $d)$, we have
	$$ T( v_{\xx} - v_{\yy} ) \leq T (v_{max}- v_{min} ) < d_{\xx} (x_0,y_0) = y_0,$$
	so $ y_0 + T v_{\yy} -  T v_{\xx}  > 0 $ for every $v_{\xx}, v_{\yy}$.
	Then,
	\eqref{proofTheo2:eq1} is equal to 
	$$ u_{\xx}(x_0,v_{\xx},y_0,v_{\yy}) = y_0 + T(  v_{\yy} - v_{\xx} ) , $$
	which is bounded by 
	$$ u_{\xx}(x_0,v_{\xx},y_0,v_{\yy}) \leq  y_0 + T(  v_{\yy} - v_{min} )  = u_{\xx}(x_0,v_{min},y_0,v_{\yy} ).$$
	
	We conclude that $v_{\xx} = v_{min}$ is $\xx$'s best reply to any strategy $v_{\yy}$ played by $\yy$.
	
	Similarly, if $\xx$ chooses strategy $v_{\xx}$, then
	\begin{equation}\label{proofTheo2:eq2}
		u_{\yy}(x_0,v_{\xx},y_0,v_{\yy}) = d_{\yy}( T v_{\xx}, y_0 + T v_{\yy} ).  
	\end{equation}
	
	We have already proven that  $y_0 + T v_{\yy} -  T v_{\xx}  > 0$ for every $v_{\xx}, v_{\yy}$, so \eqref{proofTheo2:eq2} is equal to
	$$ u_{\yy}(x_0,v_{\xx},y_0,v_{\yy}) = D + T v_{\xx} - y_0 - T v_{\yy} .$$
	
	We can bound the last expression by
	$$ u_{\yy}(x_0,v_{\xx},y_0,v_{\yy}) =   D- y_0 + T( v_{\xx} - v_{\yy})  \leq D- y_0  + T( v_{\xx} - v_{min}) = u_{\yy}(x_0,v_{\xx},y_0,v_{min} ).  $$
	
	This implies $v_{\yy} = v_{min}$ is $\yy$'s best reply to any strategy $v_{\xx}$ played by $\xx$. The conclusion is that $(v_{min}, v_{min})$ is the unique Nash equilibrium.
\end{itemize}

\begin{flushright}
	$\square$
\end{flushright}

\textbf{Proof of Theorem \ref{theo:repeat}}: \\

First, note that $d_0>d$ implies $N \equiv 0$, and the result holds trivially. 

Assume that $0< d_0 < d$, and suppose that $0< d_k < d$ for some $k \geq 0$. Then, the strategies $(U,v_{min}), (U,V), (U, v_{max} - \frac{d_k}{T} + \frac{\varepsilon}{T} ), (v_{min}, v_{min}), (v_{min}, V)$ lead to $0 <d_{k+1} < 0$ with probability one.

If the strategies of $\xx$ and $\yy$ are instead $(v_{min}, v_{max} - \frac{d_k}{T} - \frac{\varepsilon}{T} )$, then 
$d_{k+1}= d+ \varepsilon$. We can uniformly bound from below the probability that the players adopt these strategies by
$$ \mathbb{P} \Big( (X,Y)= \Big( v_{min}, v_{max} - \frac{d_k}{T} - \frac{\varepsilon}{T}  \Big) \Big)= \Big( 1 - \frac{d-d_k}{D} \Big) \Big( 1- \frac{d_k}{D} \Big) \geq \Big(1 - \frac{d}{D} \Big),  \quad \forall \enspace 0 < d_k < d,$$
where the inequality can be obtained by calculus (or by noting that this probability is an inverted parabola, as a function of $d_k$). Therefore, 
$$ \mathbb{P} (N > k ) \leq \mathbb{P} (G >k)  ,$$
where $G$ is a geometric random variable with parameter $1- \frac{d}{D}$, and the result follows. 

Finally, assume that $d_0=d$. If players $\xx$ and $\yy$ choose $(v_{min}, v_{min})$, then $d_1=d$. Any other strategy choice yields
$d_1 \neq d$.

Define $M = \min \{ n \geq 1:  d_n \neq d \}$. By the above remark, $M$ has geometric distribution on the natural numbers with parameter $ 1-  \Big(  1- \frac{2 \varepsilon}{D}  \Big) \Big( \frac{2d}{D}  \Big) $. After $M$ trials, we are on the conditional space where $\xx$ and $\yy$ do not play $(v_{min}, v_{min})$, instead they choose
\begin{eqnarray*}
	(v_{max}, v_{min} ) & & \textrm{ with probability }  \frac{\Big( \frac{2 \varepsilon}{D}  \Big) \Big( \frac{2d}{D} \Big)  }{ 1-  \Big(  1- \frac{2 \varepsilon}{D}  \Big) \Big( \frac{2d}{D}  \Big)  }, \\
	(v_{min}, v_{min} +\frac{\varepsilon}{T} ) & &  \textrm{ with probability } \frac{\Big( 1- \frac{2 \varepsilon}{D}  \Big) \Big( 1 - \frac{2d}{D} \Big)  }{ 1-  \Big(  1- \frac{2 \varepsilon}{D}  \Big) \Big( \frac{2d}{D}  \Big)  },       \\
	(v_{max}, v_{min} + \frac{\varepsilon}{T} ) & &  \textrm{ with probability } \frac{\Big(  \frac{2 \varepsilon}{D}  \Big) \Big( 1 - \frac{2d}{D} \Big)  }{ 1-  \Big(  1- \frac{2 \varepsilon}{D}  \Big) \Big( \frac{2d}{D}  \Big)  }.
\end{eqnarray*}
The first election leads to $d_{M+1}=0$, while the other two give $d_{M+1} >d$. This concludes the proof.

\begin{flushright}
	$\square$
\end{flushright}

\end{document}